\documentclass[a4paper,UKenglish,cleveref, autoref, thm-restate]{article}
\usepackage[dvipsnames]{xcolor}
\usepackage{mathtools} 
\usepackage{amsthm}
\usepackage{thmtools}
\usepackage{enumerate}
\usepackage{amsmath}
\usepackage{amssymb}
\usepackage{algorithm}
\usepackage[noend]{algpseudocode}
\usepackage{bm}
\usepackage{bbm}
\usepackage[margin=1in]{geometry}
\usepackage{amsfonts}
\usepackage{amscd}
\usepackage{enumitem}
\setlist[itemize]{noitemsep, topsep=3pt}
\usepackage{microtype}
\usepackage{tikz}
\usetikzlibrary{positioning}
\usetikzlibrary{arrows.meta}
\usetikzlibrary{calc}
\newcounter{casecounter}

\usepackage{xspace}
\usepackage{color}
\usepackage{url}
\usepackage{comment}
\usepackage[colorlinks]{hyperref}
\hypersetup{
 colorlinks={true},
  linkcolor=[rgb]{0.3,0.3,0.6},
  citecolor=[rgb]{0.2, 0.6, 0.2},
  urlcolor=[rgb]{0.6, 0.2, 0.2}
}
\usepackage[capitalise, nameinlink]{cleveref}
\usepackage{bbm}
\usepackage{times}
\usepackage{titlesec}

\newtheorem{theorem}{Theorem}[section]
\newtheorem{observation}[theorem]{Observation}
\newtheorem{lemma}[theorem]{Lemma}
\newtheorem{corollary}[theorem]{Corollary}

\newtheorem{proposition}[theorem]{Proposition}
\newtheorem{definition}[theorem]{Definition}

\newtheorem{mainthm}{Theorem}

\usepackage{authblk}
\newcommand{\uni}{University}
\newcommand{\saar}{Saarland}
\newcommand{\sic}{\saar\xspace Informatics Campus}
\newcommand{\sbn}{Saarbr\"ucken}
\newcommand{\de}{Germany}

\newcommand{\mpi}{Max Planck Institute for Informatics}

\title{The Complexity of Bisimilarity and Model Checking in Finitary Diagrams} 

 \author[1]{Markus Bl{\"a}ser\thanks{Email: \texttt{mblaeser@cs.uni-saarland.de}.}}

 \author[2]{Sagnik Dutta\thanks{Email: \texttt{sadutta@mpi-inf.mpg.de}.}}

 \author[3]{Samuel Okyay}

\affil[1,3]{\saar\xspace\uni, \sic, \sbn, \de}
\affil[2]{\mpi, \sic, \sbn, \de}

\date{}

\newcommand*\xbar[1]{%
  \hbox{%
    \vbox{%
      \hrule height 0.7pt 
      \kern0.4ex
      \hbox{%
        \kern0.05em
        \ensuremath{#1}%
        \kern0.05em
      }%
    }%
  }%
} 

\newcommand{\calA}{\mathcal{A}}
\newcommand{\calC}{\mathcal{C}}
\newcommand{\be}{\xbar{e}}
\newcommand{\Rset}{\mathbb{R}}
\newcommand{\Nset}{\mathbb{N}}

\newcommand{\leP}{\le_\mathrm{P}}
\newcommand{\leNP}{\le_{\mathrm{NP}}}
\newcommand{\leNEXP}{\le_{\mathrm{NEXP}}}
\newcommand{\poly}{\operatorname{poly}}

\newcommand{\NP}{\mathsf{NP}}
\newcommand{\coNP}{\mathsf{coNP}}
\newcommand{\RP}{\mathsf{RP}}

\newcommand{\EXPSPACE}{\mathsf{EXPSPACE}}
\newcommand{\NEXP}{\mathsf{NEXP}}
\newcommand{\PSPACE}{\mathsf{PSPACE}}

\newcommand{\existsR}{\exists\Rset}
\newcommand{\ETR}{\mathrm{ETR}}

\newcommand{\ETIM}{\mathrm{ETIM}}
\newcommand{\genETIM}{\mathrm{genETIM}}
\newcommand{\ETSM}{\mathrm{ETSLM}}
\newcommand{\genETSM}{\mathrm{genETSLM}} 
\newcommand{\SDIT}{\mathrm{SDIT}}

\newcommand{\TQBF}{\mathrm{TQBF}}

\newcommand{\FF}{\mathrm{FF}}
\newcommand{\GL}{\mathrm{GL}}
\newcommand{\posFF}{\mathrm{posFF}}

\newcommand{\BisimFD}{\mathrm{BisimFD}}
\newcommand{\rk}{\operatorname{rk}}
\newcommand{\ETRinv}{\mathrm{ETR}_\mathrm{inv}}

\newcommand{\sabo}{\operatorname{sabove}}
\newcommand{\abo}{\operatorname{above}}
\newcommand{\comm}{\operatorname{comm}}

\newcommand{\match}{\operatorname{match}}
\newcommand{\Z}{\mathbb{Z}}

\newcommand{\id}{\mathbf{I}}

\renewcommand{\de}{\overrightarrow{E}}

\newcommand{\F}{\mathbb{F}}

\begin{document}

\maketitle

\begin{abstract}
 Inspired by the work of Dubut, Goubault, and Goubault-Larrecq (ICALP 2015) on natural homology, Dubut (RAMiCS 2020) introduces finitary diagrams and studies bisimilarity and diagrammatic path logics for them. To this aim, he defines a fragment of the existential theory of the reals, called the existential theory of invertible matrices (ETIM). Using a PSPACE upper bound for this fragment, he proves that for finitary diagrams, bisimilarity can be decided in EXPSPACE and model checking for diagrammatic path logic in PSPACE. 

We significantly improve both these bounds and settle the complexity of model checking for finitary diagrams. As our first main result, we show that there is an efficient randomized algorithm for ETIM. Combining this with the previous work by Dubut, we obtain an NEXP upper bound for bisimilarity of finitary diagrams and an NP upper bound for diagrammatic path logic. We also provide a matching NP-hardness proof for the latter. The hardness proof introduces constrained layered poset problems, which may be of independent interest, and connects them to finitary diagrams using Gabriel's theorem for representations of path quivers.
For bisimilarity over finite fields, we further improve the upper bound to PSPACE.  In ETIM, we quantify over invertible matrices. We finally ask what happens if we instead quantify over matrices from the special linear group, that is, of determinant one. We show that in this case, the resulting fragment is equivalent to the existential theory of the reals, under a mild generalization of the allowed linear constraints.
\end{abstract}

\section{Introduction}

The study of behavioral equivalence in complex systems has been a central theme in concurrency theory and coalgebraic semantics. Among the most prominent notions of equivalence is bisimilarity, which provides a way to compare systems based on their observable behavior rather than their internal structure. In classical models, such as transition systems, bisimilarity relates states in such a way that every move of one system can be matched by a move of the other.  Dubut, Goubault, and Goubault-Larrecq \cite{DBLP:conf/icalp/DubutGG15} proposed a categorical version of this idea in the setting of directed algebraic topology, where systems are represented not merely by states and transitions, but by diagrams with values in algebraic categories. This approach enables the comparison of systems with rich algebraic structure, including linear dynamical systems and weighted automata.

Dubut \cite{DBLP:conf/RelMiCS/Dubut20} introduced \emph{finitary diagrams} as a finite, algebraic setting in which this categorical notion of bisimilarity can be studied algorithmically.  A finitary diagram is a functor from a finite poset to a category of finite-dimensional vector spaces. While deciding the bisimilarity of two diagrams, there are two problems: first, finding out how to relate the executions and second, constructing the 
bisimulation, in particular, the isomorphisms. The first part is difficult in general, because 
this relation is necessarily infinite when there are loops. Restricting the input category to a finite poset removes this difficulty and allows us to focus on the algebraic problem of finding suitable isomorphisms between vector spaces.  This makes finitary diagrams a natural setting for understanding the complexity of bisimilarity.

Dubut, Goubault, and Goubault-Larrecq \cite{DBLP:conf/icalp/DubutGG15} proposed a homology theory based on natural systems of abelian groups, meant to reflect directed structure, unlike classical homology which ignores direction.
A \emph{diagram} is a functor $F: C \to A$, where $C$ is a small category and $A$ is a category of observations. Inspired by the theory in \cite{joyalNielsenWinskel1996} of comparing transition systems, Dubut, Goubault, and Goubault-Larrecq \cite{DBLP:conf/icalp/DubutGG15} defined
two diagrams $F: C \to A$ and $G: D \to A$ to be bisimilar if there is a span of open morphisms between them, i.e., a diagram $H : E \to A$ and two open morphisms from $H$ to $F$ and $H$ to $G$, respectively. Dubut \cite{DBLP:conf/RelMiCS/Dubut20} studies equivalent notions of bisimilarity. 
He first defines the notion of bisimulation between two diagrams $F: C \to A$ and $G: D \to A$: A bisimulation should identify pairs of elements $c\in C$ and $d \in D$ such that for any morphism $i: c \to c'$ of $C$, there must be a morphism $j: d \to d'$ of $D$ and an isomorphism $g: F(c') \to G(d')$ satisfying the following commutativity relation (see also Figure~\ref{fig:bisim})
\begin{equation} \label{eq:com}
  g \circ F(i) = G(j) \circ f,
\end{equation}
where $F(i)$ and $G(j)$ are the induced morphisms of $A$. 
He shows that $F$ and $G$ are bisimilar if and only if there is a bisimulation between them. Then he goes on to show that bisimilarity can also be characterized in terms of path logics, in the spirit of \cite{hennessyMilner1980,joyalNielsenWinskel1996}. He defines a logic called diagrammatic path logic and proves that two diagrams are bisimilar if and only if they are logically equivalent, i.e., for every $c \in C$, there is a $d \in D$ such that for every diagrammatic path formula $S$, either $F,c$ and $G,d$ are both a model for $S$ or both are not. 

\subsection{Previous results}

Finding  bisimulations for finitary diagrams typically has a ``combinatorial'' part and an ``algebraic'' part. In the combinatorial part, we have to find an alignment of the objects in $C$ and $D$  whereas in the algebraic part, we need to find the isomorphisms $f$ and $g$. 

As a first tool for computing these isomorphisms, \cite{DBLP:conf/RelMiCS/Dubut20} introduced the existential theory of invertible matrices ($\ETIM_{\Rset})$.  It contains sentences of the form
\[
   \exists_{n_1} X_1 \exists_{n_2} X_2 \dots \exists_{n_k} X_k: \bigwedge_{i = 1}^m P_i(X_1,\dots,X_k). 
\]
Here $n_\kappa \ge 0$ are natural numbers, $1 \le \kappa  \le k$, and $X_\kappa$ are variables that quantify over invertible matrices in $\Rset^{n_\kappa \times n_\kappa}$. $P_j$ is a predicate of the form 
\begin{equation} \label{eq:etim}
     A X_\kappa = X_\mu B \quad \text{for $1 \le \kappa , \mu \le k$}
\end{equation} 
for some matrices $A$ and $B$ of matching sizes and with rational entries. Such systems of equations can be used to model the commutativity relations in \cref{eq:com}. It is easy to see that $\ETIM_{\Rset}$ is a fragment of the existential theory of the reals ($\ETR$). Since $\ETR \in \PSPACE$ \cite{existentialTheoryOfRealsCanny1988some,renegar1992computational}, this implies that $\ETIM_{\Rset} \in \PSPACE$ too. The survey \cite[L-Open2]{DBLP:journals/corr/abs-2407-18006} asks the natural question whether $\ETIM_{\Rset}$ is $\existsR$-complete.

Utilizing the $\PSPACE$ upper bound for $\ETIM_{\Rset}$, Dubut \cite{DBLP:conf/RelMiCS/Dubut20} shows that bisimilarity of finitary diagrams can be decided in $\EXPSPACE$. For this, he guesses the tuples of a bisimulation relation with placeholders for the isomorphisms, implicitly using that if there is  bisimulation, then there is one of at most exponential size. Then to check whether isomorphisms in the guessed relation can be instantiated such that they fulfill the commutativity relations in (\ref{eq:com}), he sets up a system of equations and uses the $\PSPACE$ algorithm for $\ETIM$ on an exponentially large instance to check its feasibility. 

In a similar fashion, Dubut \cite{DBLP:conf/RelMiCS/Dubut20} solves the model checking problem for positive diagrammatic path logic in $\PSPACE$. He first guesses the ``combinatorial'' part of the model checking problem, then sets up a system of equations of the form (\ref{eq:com}), and then invokes the $\PSPACE$ algorithm for $\ETIM_{\Rset}$. 

\subsection{Our contributions}

We begin by showing that the aforementioned algebraic part of finding bisimulations is often easy. In particular, we demonstrate that the existential theory of invertible matrices $\ETIM$ allows for an efficient randomized algorithm over the reals, i.e., $\ETIM_{\Rset} \in \RP$, greatly improving on the upper bound of $\PSPACE$ by \cite{DBLP:conf/RelMiCS/Dubut20}. This is still true when we allow arbitrary linear constraints in the entries of the matrices and not only constraints of the form (\ref{eq:etim}). We call this generalization $\genETIM_{\Rset}$. We obtain our efficient algorithm by reducing the problem to the well-known \emph{polynomial identity testing problem} (PIT).  \cite[L-Open2]{DBLP:journals/corr/abs-2407-18006} asks whether $\ETIM_{\Rset}$ is $\existsR$-complete. Our results answer this question in the negative (assuming $\RP \not= \existsR$). We also show that derandomizing the algorithm  will be difficult, at least for $\genETIM_{\Rset}$, since this would be equivalent to derandomizing symbolic determinant identity testing ($\SDIT_{\Rset}$), which is a major open problem in complexity theory \cite{DBLP:journals/cc/KabanetsI04}.

$\ETIM_{\Rset} \in \RP$ implies that $\ETIM_{\Rset} \in \NP$. In fact, we are able to show that $\ETIM_{\F} \in \NP$ for all fields $\F$. Using this new upper bound for $\ETIM$ readily gives improved upper bounds for testing bisimilarity of finitary diagrams as well as model checking for (negation-free) diagrammatic path logic of finitary diagrams. We call these problems $\BisimFD_{\F}$ and $\posFF_{\F}$ respectively. We get that $\posFF_{\F} \in \NP$ and $\BisimFD_{\F} \in \NEXP$ for all fields $\F$.

For $\posFF_{\F}$, we prove a matching lower bound --- we show that the problem is $\NP$-hard too. In the reduction, we define constrained layered poset problems, which might be of independent interest for showing hardness proofs in the context of finitary diagrams. To relate constrained layered posets to finitary diagrams, we make use of Gabriel's theorem for the representation of path quivers. For $\BisimFD_{\F}$, we go further and present a $\PSPACE$ upper bound when $\F$ is a finite field.

Finally, we study the theory of special linear matrices. Instead of quantifying over invertible matrices, here we quantify over matrices with determinant $1$. This corresponds to isomorphisms that are volume preserving. While we show that the generalized theory of real invertible matrices is in $\RP$, we prove that the generalized theory of special linear matrices is $\existsR$-complete. We leave it as an open question whether this is also true if we only allow for linear constraints of the form $(\ref{eq:etim})$.

\section{Preliminaries}

\subsection{Finitary diagrams}

\begin{definition}[\cite{DBLP:conf/RelMiCS/Dubut20}] A \emph{finitary diagram} $F$ over a field $\F$ consists of the following data:
    \begin{enumerate}
        \item a finite partially ordered set (for short poset) $(C,\leq)$ which describes the domain,
        \item for every $c \in C$, a natural number $F(c)$ (which stands for the vector space $\F^{F(c)}$),
        \item for every pair $c\leq c'$ of $C$, a matrix $F(c\leq c')$ of size $F(c') \times F(c)$, with coefficients from $\F$, such that:
        \begin{itemize}
            \item $F(c\leq c)$ is the identity matrix for all $c \in C$,
            \item for every triple $c\leq c'\leq c''$, $F(c\leq c'') = F(c'\leq c'') \cdot F(c \leq c')$, where ``$\cdot$'' denotes matrix multiplication.
        \end{itemize}
    \end{enumerate}
\end{definition}

\subsection{Bisimilarity}

We first give the general definition of bisimulations in the setting of \cite{DBLP:conf/icalp/DubutGG15} and \cite{DBLP:conf/RelMiCS/Dubut20} and then specialize it to finitary diagrams.

\begin{definition}
A \emph{bisimulation} $R$  between two diagrams $F : C \rightarrow A$ and $G:D \rightarrow A$ is a set of triples $(c,f,d)$ where $c$ is an object of $C$, $d$ is an object of $D$ and $f: F(c) \rightarrow G(d)$ is an isomorphism of A such that:
    \begin{enumerate}
        \item For every $(c,f,d) \in R$ and $i: c \rightarrow c' \in C$, there exists $j: d\rightarrow d' \in D$ and $g: F(c') \rightarrow G(d') \in A$ such that $g \circ F(i) = G(j) \circ f$ and $(c',g,d') \in R$, see \cref{fig:bisim}.
        \item Symmetrically, for every $(c,f,d) \in R$ and $j:d \rightarrow d' \in D$, there exists $i:c \rightarrow c' \in C$ and $g : F(c') \rightarrow G(d') \in A$ such that $g \circ F(i) = G(j) \circ f$ and $(c',g,d')\in R$.
        \item $\forall c \in C: \exists d,f: (c,f,d) \in R$.
        \item $\forall d \in D: \exists c,f: (c,f,d) \in R$.
    \end{enumerate}
\end{definition}

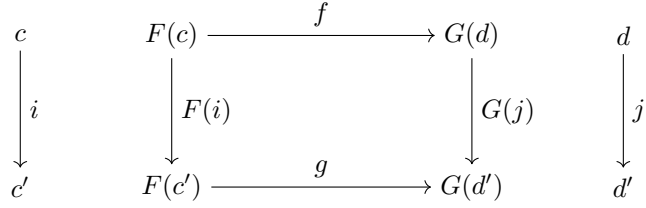
\begin{figure}
\begin{center}
  \begin{tikzpicture}[main/.style = {draw, circle}, node distance=2cm and 2cm] 
        \node (1) {$c$};
        \node [below of = 1](2) {$c'$};
        \node [right of = 1](3) {$F(c)$};
        \node [below of = 3](4) {$F(c')$};
        \node [right = 3cm of 3](5) {$G(d)$};
        \node [below of = 5](6) {$G(d')$};
        \node [right of = 5](7) {$d$};
        \node [below of = 7](8) {$d'$};
        
        \draw[->] (1) --  node[auto] {$i$} (2);
        \draw[->] (3) --  node[auto] {$F(i)$} (4);
        \draw[->] (5) --  node[auto] {$G(j)$} (6);
        \draw[->] (3) --  node[auto] {$f$} (5);
        \draw[->] (4) --  node[auto] {$g$} (6);
        \draw[->] (7) --  node[auto] {$j$} (8);
    \end{tikzpicture}  
\end{center}
\caption{If $f$ identifies $c$ with $d$ and $i$ is a morphism $c \to c'$, then there must be an object $d'$ and an isomorphism $g$ such that the diagram commutes.   \label{fig:bisim}}
\end{figure}

\begin{definition} \label{def:bisim}
    Two diagrams are called \emph{bisimilar} if there is a bisimulation between them.
\end{definition}

We will denote the problem of bisimilarity testing in finitary diagrams by $\BisimFD_{\F}$, where $\F$ is the underlying field of the vector spaces in the diagram. The main result of \cite{DBLP:conf/RelMiCS/Dubut20} on the bisimilarity of finitary diagrams is the following.
\begin{proposition}
    $\BisimFD_{\Rset} \in \EXPSPACE$.
\end{proposition}

\subsection{Diagrammatic path logic and finitary formulae}
\label{subsec:FF}

\cite{DBLP:conf/RelMiCS/Dubut20} introduces the so-called \emph{diagrammatic path logic}, which is similar to the logic
introduced by \cite{hennessyMilner1980} for transition systems or to path logics developed by \cite{joyalNielsenWinskel1996}. \emph{Finitary formulae} are an instance of diagrammatic path logic for finitary diagrams. The syntax is as follows.
\begin{description}
\item[Object formulae:] $S ::= [n] P$ with $n \in \Nset$
\item[Morphism formulae:] $P ::= \langle M \rangle P\ |\ ?S\ |\ \neg P\ |\  P_1 \wedge P_2\ |\ \top $,
\end{description}
where $M$ is a matrix over some field $\F$. Here, $[n]P$ asserts that the current object represents the vector space $\F^n$,
$\langle M\rangle P$ ``fires'' a transition via the matrix $M$,
``?'' transitions back to object-level evaluation, $\neg$ and $\wedge$ are standard Boolean connectives, and
$\top$ is the tautology.

The semantics are as follows: For a diagram $F: \calC \to \calA$, an object $c \in \calC$, and an isomorphism $f$ of $\cal A$
of the form $f : \F^{F(d)} \to \F^{F(d)}$ for some $d$, we define $F, c \models S$ for an object formula $S$, and $F, f, d \models P$ for a morphism formula $P$ by induction on the structure:
\begin{enumerate}
    \item $F,c \models [n]P$ iff $F(c) = n$ and $F,f,c \models P$ for some isomorphism $f: \F^{F(c)} \to \F^{F(c)}$.
    \item $F,f,c \models \langle M \rangle P$ iff there is a $c \le c'$ in $\cal C$ and an isomorphism $f'$ of $\F^{F(c')}$ such that $M f = f' F(c \le c')$.
    \item ? switches back to object formulae, $F, f, c \models ?S$
    iff $F,c \models S$.
    \item Conjunction has the usual semantics, i.e., $F, f, c \models P_1 \wedge P_2$ iff $F, f, c \models P_1$ and $F, f, c \models P_2$.
    \item The same is true for negation, $F, f, c \models \neg P$ iff $F, f, c \not \models P$.
    \item Finally, $\top$ is always satisfied, i.e., $F, f, c \models \top$ always holds.
\end{enumerate}
This setup closely mirrors labeled transition or path logics, but here the ``labels'' are matrices over a field.

Let $\FF$ denote all triples $(F,c,S)$, where $S$ is a finitary object formula, such that $F,c \models S$. A finitary formula is called \emph{positive} if it does not contain any negations. Let $\posFF_{\F}$  denote the subset of $\FF$ corresponding to positive finitary formulas, with $\F$ being the underlying field. \cite[Thm. 9]{DBLP:conf/RelMiCS/Dubut20} shows the following:
\begin{proposition}
    $\posFF_{\Rset} \in \PSPACE$.
\end{proposition}

\subsection{Existential theory of the reals}

The \emph{existential theory of the reals} ($\ETR$) is the decision problem of determining the truth of formulas of the form
\begin{equation}
\exists x_1, \ldots, x_n \; : \;
\Big(p_1(x_1, \ldots, x_n) \ \bowtie_1 \ 0\Big) \;\land\; \cdots \;\land\;
\Big(p_m(x_1, \ldots, x_n) \ \bowtie_m \ 0\Big), \label{eq:etr}
\end{equation}
where each $p_i$ is a multivariate polynomial with integer (or rational) coefficients, and
\[
\bowtie_i \;\in\; \{=,\ <,\ \le,\ >,\ \ge\}.
\]
The complexity class $\existsR$ is the set of all languages polynomial-time many-one reducible to the $\ETR$ problem. It satisfies
 $\NP \subseteq \existsR \subseteq \PSPACE$.
By now, there is an abundance of complete problems for $\existsR$ known, see the recent compendium \cite{DBLP:journals/corr/abs-2407-18006}.

\subsection{Existential theory of invertible matrices}

Dubut \cite{DBLP:conf/RelMiCS/Dubut20} defines the \emph{existential theory of invertible matrices} as an intermediate problem. It contains sentences of the form
\[
   \exists_{n_1} X_1 \exists_{n_2} X_2 \dots \exists_{n_k} X_k: \bigwedge_{i = 1}^m P_i(X_1,\dots,X_k). 
\]
Here $n_\kappa \ge 0$ are natural numbers, $1 \le \kappa  \le k$ and $X_\kappa$ are variables that quantify over invertible matrices in $\F^{n_\kappa \times n_\kappa}$, for some field $\F$. $P_j$ is a predicate of the form $A X_\kappa = X_\mu B$ for $1 \le \kappa , \mu \le k$ for some matrices $A$ and $B$ of matching sizes and with entries from $\F$. $\ETIM_{\F}$ is the set of all true sentences of the above form. 

The predicates were chosen to be of the above form because they naturally appear in the case of finitary diagrams. We can also consider a more general problem where each $P_i$ is an arbitrary affine linear equation in the entries of the matrices $X_1,\dots,X_k$. We call this problem $\genETIM_{\F}$, as it is a generalization of $\ETIM_{\F}$. 

\subsection{Nondeterministic reductions}

As a tool to prove the containment of problems in $\NP$ or $\NEXP$, we will use \emph{nondeterministic reductions}, which already implicitly appear in \cite{DBLP:conf/RelMiCS/Dubut20}.

\begin{definition}
    A language $A$ is nondeterministically polynomial time many-one reducible to $B$ if there is a deterministically polynomial time computable function $f$ with two inputs such that for all $x$: $x \in A$ iff there is a $y$ with $|y| \le \poly(|x|)$ such that $f(x,y) \in B$. We write $A \leNP B$.
\end{definition}

\begin{proposition} \label{prop:reduc:NP}
    If $A \leNP B$ and $B \in \NP$, then $A \in \NP$.
\end{proposition}

We will also need nondeterministic exponential-time reductions. Exponential here means $2^{\poly(n)}$.

\begin{definition}
    A language $A$ is nondeterministically exponential-time many-one reducible to $B$, if there is an exponential time computable function $f$ with two inputs such that for all $x$: $x \in A$ iff there is a $y$ with $|y| \le 2^{\poly(|x|)}$ such that $f(x,y) \in B$. We write $A \leNEXP B$.
\end{definition}

\begin{proposition} \label{prop:reduc:NEXP}
    If $A \leNEXP B$ and $B \in \NP$, then $A \in \NEXP$.
\end{proposition}

\subsection{A tool from quiver theory and persistent homology}

A finitary diagram is a functor from a poset category to the category of finite-dimensional vector spaces and linear maps. This is closely related to the notion of representations of a quiver and when the poset is totally ordered, it is exactly the same as single-parameter persistent modules. Therefore, tools from quiver theory and persistent homology can prove useful for problems on finitary diagrams. We use one particular tool which follows from the well-known Gabriel's theorem on quivers and also appears in persistent homology as the rank invariant criterion \cite[Theorem 12]{persistence}.

\begin{theorem}\label{rank-invariant}
    We are given the matrices $A_1, \cdots, A_k$ and $B_1, \cdots, B_k$ over some field $\F$, where for each $i \in [k]$, the matrices $A_i$ and $B_i$ have dimension $d_{i+1} \times d_i$ for some integers $d_1,\cdots,d_{k+1}$. For $1 \leq i \leq i' \leq k$, let $A_{[i,i']}$ denote the product $A_{i'}A_{i'-1}\cdots A_{i}$ and $B_{[i,i']}$ denote the product $B_{i'}B_{i'-1}\cdots B_{i}$. Then,
    \begin{align*}
        &\exists \text{ invertible matrices } X_1,\cdots,X_{k+1}\ :\ \forall i \geq 2,\ X_i\cdot A_{i-1} = B_{i-1} \cdot X_{i-1}\\
        &\iff \text{ for all } 1 \leq i \leq i' \leq k,\ \rk A_{[i,i']} = \rk B_{[i,i']}\ .
    \end{align*}
\end{theorem}

This criterion follows from the classification of representations of the Dynkin quiver $A_n$ (a special case of Gabriel's theorem). For more background on the criterion and its connection to Gabriel's theorem, refer to \cref{gabriel}.

\section{Overview of results and techniques}

We give a comprehensive overview of our results and explain the main techniques used in our proofs.

\subsection{Existential theory of invertible matrices}

Dubut introduces the existential theory of invertible matrices ($\ETIM$) to get upper bounds for deciding bisimilarity of finitary diagrams and model checking of finitary formulae. An instance of $\ETIM_{\F}$ is of the form $\exists_{n_1} X_1 \exists_{n_2} X_2 \dots \exists_{n_k} X_k: \bigwedge_{i = 1}^m P_i(X_1,\dots,X_k)$, where each constraint is of the form $A X_i = X_j B$ for matrices $A$ and $B$ over $\F$. If we consider the case of the real field, an instance of $\ETIM$ can be easily translated into an equivalent instance of the existential theory of the reals. This leads \cite[L-Open2]{DBLP:journals/corr/abs-2407-18006} to ask the natural question whether $\ETIM_{\Rset}$ is $\existsR$-complete. We answer this question in the negative (assuming $\RP \not= \existsR$) by giving an efficient randomized algorithm.

Dubut uses $\ETIM_{\F}$ to verify the commutativity relations of the form $(\ref{eq:etim})$, which explains the structure of the constraints $P_i$. It turns out that our algorithm can also handle the case of arbitrary affine linear equations in the entries of the matrices, which we call $\genETIM_{\F}$.

\begin{mainthm}[\cref{thm:genETIM:RP}]
  $\genETIM_{\F} \in \RP$ for infinite fields $\F$.
\end{mainthm}

\begin{proof}[Proof overview] In a $\genETIM$ instance, we have linear equations from the constraints $P_1,\dots,P_m$, and polynomial inequalities involving the determinants, expressing that the matrices $X_1,\dots,X_k$ are invertible. First, we parametrize the solution space of the linear system using free variables and substitute this parametrization into the matrix variables. The instance is satisfiable if and only if after the substitution, all of the determinants are nonzero polynomials. This can be tested with the famous Schwartz-Zippel lemma. We need the infiniteness of the field to sample enough points from it for the use of the lemma.
\end{proof}

One can ask the question whether our algorithm can be derandomized. This turns out to be a hard problem, since it is equivalent to the complement of the symbolic determinant identity problem, whose derandomization over the rationals/reals in particular implies strong circuit lower bounds (see \cite{DBLP:journals/cc/KabanetsI04}). 

\begin{mainthm}[\cref{thm:genETIM:SDIT}]
     $\genETIM_{\F}$ is deterministically polynomial time equivalent to the complement of the symbolic determinant identity testing problem $\SDIT_{\F}$ for all fields $\F$.
\end{mainthm}

\begin{proof}[Proof overview] The reduction from $\genETIM_{\F}$ to the complement of $\SDIT_{\F}$ is already implicit in the proof strategy of the above theorem.
For the other direction, we essentially use the linear constraints of $\genETIM_{\F}$ to specify the affine linear entries of the $\SDIT_{\F}$ instance. 
\end{proof}

\subsection{First results through our algorithm for ETIM}

Dubut essentially constructs nondeterministic reductions from $\BisimFD$ and $\posFF$ to $\ETIM$. The first one is an exponential time reduction, the second one is polynomial-time. In both reductions, he nondeterministically guesses the assignments between states and for each such guess, he creates an equation of the form (\ref{eq:etim}) to check the existence of a matching isomorphism. All these checks can be pushed to the end, making the algorithms by Dubut essentially nondeterministic reductions. Since $\ETIM_{\F} \in \NP$ (using \cref{thm:genETIM:RP} for infinite fields and trivially for finite fields) and $\NP$ is closed under nondeterministic polynomial time reductions and the closure of $\NP$ under nondeterministic exponential time reductions is $\NEXP$, we get the following results

\begin{mainthm}[\cref{cor:Bisim:NEXP}]
$\BisimFD_{\F} \in \NEXP$ for all fields $\F$.
\end{mainthm}

\begin{mainthm}[\cref{cor:posFF:NP}]
$\posFF_{\F} \in \NP$ for all fields $\F$.
\end{mainthm}

The upper bound for $\posFF$ is optimal, the one for $\BisimFD$ can be further improved for finite fields.

\subsection{Complexity of bisimilarity checking for finitary diagrams}

\begin{mainthm}[\cref{alt:bisim:PSPACE}]
    $\BisimFD_{\F} \in \PSPACE$ when $\F$ is a finite field.
\end{mainthm}

\begin{proof}[Proof overview]
The above reduction approach to $\ETIM$ produces a system of equations of exponential size, therefore, we have to use a different approach. We set up a quantified formula that is true iff the given diagrams are bisimilar. This formula quantifies over Boolean variables and matrices over $\F$. Therefore, we can brute-force over all possibilities in $\PSPACE$. 
\end{proof}

\subsection{Complexity of model checking for finitary formulae}
We already showed that $\posFF$ is in $\NP$. Now we establish its $\NP$-completeness.
\begin{mainthm}[\cref{posFF-hardness}]
    $\posFF_{\F}$ is $\NP$-hard for all fields $\F$.
\end{mainthm}

\begin{proof}[Proof overview]
    We reduce the classic \textsf{CLIQUE} problem to $\posFF$. Given an undirected graph $G$ and a parameter $k$, we have to construct a finitary diagram and a finitary formula over the diagram such that the  formula is satisfiable if and only if $G$ has a clique of size $k$. A finitary diagram $F$ is a functor from a poset to the category of finite-dimensional $\F$-vector spaces and linear maps. Posets can be viewed as transitive and reflexive directed acyclic graphs, which is already a very restricted class of graphs. Further, whenever there is a chain $a \leq b \leq c$ in the poset, the diagram must satisfy $F(b \leq c)\cdot F(a\leq b) = F(a \leq c)$. All these restrictions make the reduction extremely tricky. Therefore, we construct intricate gadgets called constrained layered posets, which help us build the necessary finitary diagram. Then we construct a finitary object formula of the form $[q]\langle M_1\rangle\langle M_2\rangle \cdots \langle M_k\rangle$. The advantage of this special form is that we can use the rank invariant criterion in \cref{rank-invariant} to characterize the satisfiability of such formulas. Finally, we can ensure that the rank-invariance conditions are satisfied if and only if $G$ has a $k$-clique.
\end{proof}

\subsection{Existential theory of special linear matrices}

Instead of taking arbitrary isomorphisms for identifying the elements of the diagram, we could also consider special linear maps, that is, matrices of determinant one. This would put stronger geometric conditions on the similarity, for instance, volumes being preserved. It is natural to explore the complexity of the corresponding problems. We prove that the existential theory of special linear matrices is $\existsR$-complete, in contrast to $\genETIM_{\Rset}$. 
\begin{mainthm}[\cref{cor:genETSM}]
    $\genETSM_{\Rset}$ is $\existsR$-complete.
\end{mainthm}

\begin{proof}[Proof overview]
The reduction is gadget-based. We start from a special case of $\ETR$, called $\ETRinv$, where we are only allowed to use equations of the form  $x = 1$, $x + y = z$, and $xy = 1$, cf.\ \cite{DBLP:journals/jacm/AbrahamsenAM22}. For every variable, we set up a $2 \times 2$-matrix and use linear equations such that the matrices have the form $\left( \begin{smallmatrix}
  x & 0 \\ 0 & x'  
\end{smallmatrix} \right)$. Together with the fact that we quantify over special linear matrices, this enforces $x x' = 1$, that is, $x' = x^{-1}$. This automatically also implements equations of the form $xy = 1$. The tricky part is to implement the additions. This requires a series of cleverly chosen linear equations. In each step, we have to ensure that we do not constrain the matrices too much, since we always have to ensure that there is still a solution in which the determinant of the matrices in the equations is one.
\end{proof}


\section{An efficient algorithm for the existential theory of invertible matrices}
\label{sec:ETIM:RP}

In this section, we present our first main result, an efficient algorithm for the generalized existential theory of invertible matrices. The main insight for designing our algorithm is that we can reduce the existential theory of invertible matrices to polynomial identity testing (PIT). Then we will use the famous \emph{Schwartz-Zippel} Lemma, see e.g.\ \cite{DBLP:journals/fttcs/ShpilkaY10}, to get an efficient randomized algorithm. 

\begin{lemma}[Schwartz-Zippel]
Let $P \in R[x_1,x_2,\dots,x_n]$ be a non-zero polynomial of total degree $d>0$ over an integral domain $R$. Let $S$ be a finite subset of $R$ and let $r_1,r_2, \dots, r_n$ be selected independently and uniformly at random from $S$. Then:
\[
  \Pr[P(r_1,r_2,\dots,r_n) = 0] \leq \frac{d}{|S|}.
\]  
\end{lemma}

Let 
\begin{equation}
   \exists_{n_1} X_1 \exists_{n_2} X_2 \dots \exists_{n_k} X_k: \bigwedge_{i = 1}^m P_i(X_1,\dots,X_k). \label{eq:genETIM} 
\end{equation}
be the given $\genETIM$-instance. Recall that we quantify over invertible matrices and that $P_1,\dots,P_m$ are affine linear equations in the entries $x_{i,j}^{(h)}$ of the matrices $X_h$, $1 \le i,j \le n_h$, $1 \le h \le k$. The following algorithm decides whether the instance is true:

\begin{algorithm}
\caption{$\genETIM$ by identity testing} \label{alg:genETIM}
\hspace*{\algorithmicindent} \textbf{Input:} A $\genETIM$-instance like in \cref{eq:genETIM}\\
\hspace*{\algorithmicindent} \textbf{Output:} Whether the instance is satisfiable
\begin{algorithmic}[1]
    \State Check whether the affine system $P_1,\dots,P_m$ has a solution. If not, return 0.
    \State Compute equations of the solution space of the form $y_i = L_i(\bar Y)$, $1 \le i \le r$, where $r$ is the codimension of the solution space and $y_1,\dots,y_r$ are entries of the matrices $X_1,\dots,X_k$ and $L_1(\bar Y),\dots,L_(\bar Y)$  are affine linear forms in the remaining variables. 
    \State Let $X_h'$, $1 \le h \le k$ be the matrices obtained by replacing each $y_i$ by the corresponding right hand side $L_i(\bar Y)$.
    \State Check whether $\det(X_h')$ is not the zero polynomial, $1 \le h \le k$, using the Schwartz-Zippel lemma. If all polynomials are nonzero, return 1. Otherwise return 0. 
\end{algorithmic}
\end{algorithm}

\begin{theorem} \label{thm:genETIM:RP}
    $\genETIM_{\F} \in \RP$ for infinite fields $\F$.
\end{theorem}

\begin{proof}
We need to prove the correctness of Algorithm~\ref{alg:genETIM}. If the algorithm returns 1, then by construction it has found an assignment to $X_1,\dots,X_k$ such that the linear constraints are satisfied and each determinant is nonzero, that is, the matrix is invertible. If on the other hand the algorithm returns $0$, then either the linear system has no solution, or one of the determinants is identically zero, or one of the identity tests erroneously failed. In the first case, there is indeed no solution, since already the linear system without any invertibility constraints is not satisfiable. In the second case, there is no solution, too, since we computed the solution space of the affine system of linear constraints and one of the determinants vanishes on this space. In the third case, by choosing the set $S$ in the Schwartz-Zippel lemma large enough, we can ensure that the error probability of one test failing is $\le 1/(2k)$. This means that by the union bound, the error probability in the yes-case is bounded by $1/2$, thus we satisfy the acceptance condition of $\RP$.

The algorithm can be implemented in randomized polynomial time, since we only solve systems of 
linear equations and evaluate determinants. This proves the theorem.
\end{proof}

\begin{corollary}\label{etim-rp}
    $\ETIM_{\F} \in \RP$ for infinite fields $\F$.
\end{corollary}


Since $\RP \subseteq \NP$, we have $\genETIM_{\F} \in \NP$ for all infinite fields $\F$. On the other hand, for finite fields $\F$, we trivially have $\genETIM_{\F} \in \NP$, since one can guess the solution to the $\genETIM$-instance non-deterministically. The same is true for $\ETIM$. Therefore, we have the following corollary.

\begin{corollary}\label{etim-np}
    $\genETIM_{\F} \in \NP$ and $\ETIM_{\F} \in \NP$ for all fields $\F$.
\end{corollary}

Next, we show that our upper bound for $\genETIM$ is optimal in the sense that derandomizing it over rational/real fields would have dramatic consequences in complexity theory. The \emph{symbolic determinant identity problem} ($\SDIT_{\F}$) is the following problem: Given a square matrix $A(\mathbf{y})$ whose entries are affine linear polynomials from $\F[y_1,\cdots,y_m]$, decide whether $\det A = 0$.

\begin{theorem} \label{thm:genETIM:SDIT}
    $\genETIM_{\F}$ is deterministically polynomial time equivalent to the complement of $\SDIT_{\F}$ for all fields $\F$.
\end{theorem}

\begin{proof}
    The reduction from $\genETIM_{\F}$ to the complement of $\SDIT_{\F}$ is exactly the construction in Algorithm~\ref{alg:genETIM}. Since we stop just before invoking the Schwartz-Zippel lemma, this part of the construction holds over all fields $\F$.

    For the other direction, first we can use a standard reduction from $\SDIT$ for matrices with affine linear entries to $\SDIT$ for matrices with entries which are only variables or constants. Let $A(\mathbf{y}) = A_0 + \sum_{k=1}^m A_ky_k$ be the $n\times n$ input matrix for which we have to decide whether $\det A(\mathbf{y}) = 0$. Here, $A_k$ is a constant matrix for all $k$. Let $N = n^2m$. We define two matrices $U \in \F^{n \times N}$ and $V \in \F[\mathbf{y}]^{N \times n}$, where the columns of $U$ and the rows of $V$ are indexed by triples $(i,j,k) \in [n]^2 \times [m]$. For all $l \in [n]$ and $t = (i,j,k) \in [n]^2\times [m]$, define 
    $$U_{l,t} = \begin{cases}
        -(A_k)_{ij} & \text{if } i=l,\\
        0 & \text{ otherwise},
    \end{cases} \quad\text{ and }\quad  V_{t,l} = \begin{cases}
        y_k & \text{if } j=l,\\
        0 & \text{ otherwise}.
        \end{cases}\ .$$
    Form the matrix $B(\mathbf{y}) = \begin{bmatrix}
        A_0 & U\\
        V & I_N
    \end{bmatrix}$. Taking the Schur complement with respect to the bottom-right block, we get 
    $$\det B = \det(I_N)\det(A_0 - UI_N^{-1}V) = \det(A_0 - UV) = \det A\ .$$
    Since the entries of $B$ are only variables or constants, we have reduced the general $\SDIT$ to $\SDIT$ for these special kinds of matrices.
    
    Now, given an $\SDIT$-instance $A = (a_{i,j})$ of size $n \times n$ with entries that are variables or constants, we create a $\genETIM$-instance as follows: We quantify over one matrix $X = (x_{i,j})$ of size $n \times n$. If $a_{i,j}$ is a constant, then we add the equation $x_{i,j} = a_{i,j}$ to the instance. For each variable $y_\ell$ that appears in $A$, we let $(i_1,j_1),\dots,(i_k,j_k)$ be the entries of $A$ in which it occurs. Then we add the equations $x_{i_s,j_s} = x_{i_{s+1},j_{s+1}}$, $1 \le s < k$, to the instance. By construction $\det A \not= 0$ iff there is an invertible matrix that satisfies the constructed equations.
\end{proof}

Over finite fields, $\SDIT$ is $\coNP$-complete. Over rational/real fields, derandomizing $\SDIT$ (as well as PIT in general)   is a major open problem in computational complexity, in particular, it implies strong circuit lower bounds (\cite{DBLP:journals/cc/KabanetsI04}). 


\section{An improved algorithm for bisimilarity}

In this section, we present an improved algorithm for bisimilarity testing of finitary diagrams. \cite[Theorem 8]{DBLP:conf/RelMiCS/Dubut20} shows that this problem is in $\EXPSPACE$. We get an improvement by using our new result for $\genETIM$ (\cref{thm:genETIM:RP}).
\cite[Section 7]{DBLP:conf/RelMiCS/Dubut20} gives an algorithm, which implicitly constructs a nondeterministic reduction from testing bisimilarity of finitary diagrams to $\ETIM$:

\begin{proposition}[implicit in \cite{DBLP:conf/RelMiCS/Dubut20}] $\BisimFD_{\F} \leNEXP \ETIM_{\F}$ for all fields $\F$.
\end{proposition}

In his reduction, Dubut essentially guesses a bisimulation, that is, triples of the form $(a,X,b)$, implicitly using the fact that if there is a bisimulation, then there is always one of exponential size. The reduction works as follows: In the triple $(a,X,b)$, $a$ is from the poset $C$ of the first diagram $F: C \to A$ and $b$ is from the poset $D$ of the second diagram $G: D \to A$. $X$ is a ``placeholder'' for the isomorphism between $F(a)$ and $G(b)$. We then list all the $\ETIM$-equations that need to be satisfied according to \cref{def:bisim} and use an $\ETIM$-solver to check whether the system is feasible. 

This implicitly uses the following lemma. Let $S$ be a bisimulation between two finitary diagrams $F: C \to A$ and $G: D \to A$. We define a partial order on the set of triples in $S$ as follows:
\[
  (a, f, b) \le (a',f',b') \iff \text{$a \le a'$ and $b \le b'$ and $f' \cdot F(a \le a')  = G(b \le b') \cdot f$}.
\]

\begin{lemma}
Any bisimulation $S$ between two finitary diagrams $F: C \to A$ and $G: D \to A$ contains a subset $S' \subset S$ such that $S'$ is a bisimulation and $|S'| = 2^{\poly(n)}$, where $|C|, |D| \le n$.     
\end{lemma}

\begin{proof}
Consider $S$ with the order $\le$ defined above. Among all minimal elements in $S$, choose at most one tuple of the form $(a,f,b)$ for each pair $a \in C$ and $b \in D$. The number of tuples is bounded by $n^2$. For every tuple $(a,f,b)$ that was chosen in the first round and every $a < a'$, we choose a tuple $(a',f',b') \in S$ such that $b \le b'$ and $f' \cdot F(a \le a')  = G(b \le b') \cdot f$. Such tuples exist by the definition of bisimulation. We do the same for every $b < b'$. In this way, we add $2n$ new tuples to $S'$ for each tuple added in the first round, so we add $\le n^2 \cdot (2n)$
tuples in total in the second round. Now we go on inductively: For each tuple $(a,f,b)$ added in the previous round and each $a \le a'$, we choose one tuple $(a',f',b') \in S$ such that $f' \cdot F(a \le a')  = G(b \le b') \cdot f$ and add it to $S'$. We do the same for every $b \le b'$. This process comes to an end after $2n$ rounds. Therefore, the total size of $S'$ is bounded by $|S'| \le n^2 \cdot \sum_{i=0}^{2n} (2n)^i = 2^{\poly(n)}$. By construction, $S'$ satisfies the conditions in the definition of bisimulation.
\end{proof}

Together with our improved upper bound for $\ETIM$ (\cref{etim-np}) and \cref{prop:reduc:NEXP}, we get an improved upper bound for testing bisimilarity of finitary diagrams.

\begin{corollary} \label{cor:Bisim:NEXP}
    $\BisimFD_{\F} \in \NEXP$ for all fields $\F$.
\end{corollary}

\section{A PSPACE upper bound for bisimilarity over finite fields}

Next, we further improve the upper bound for $\BisimFD_{\F}$ to $\PSPACE$, when $\F$ is a finite field. Since the system of equations that is generated in the above reduction is of exponential size, the reduction approach to $\ETIM$ will not work unless we could prove that there is a smaller system. Instead, we will (deterministically) polynomial-time reduce this problem to $\TQBF$, the set of all true quantified Boolean formulas, which is a classical $\PSPACE$-complete problem. 

Let $F: C \to A$ and $G: D \to A$ be two finitary diagrams. We construct a quantified formula $\Phi$, which is true iff $F$ and $G$ are bisimilar:
\begin{equation} \label{eq:Phi}
  \begin{array}{ll}
    \forall a_1\quad  \exists b_1 & \exists_{D(a_1)} X_1:  \match(a_1,b_1) \, \wedge \\
    \forall a_2:\ \sabo(a_2,a_1,b_1) \implies  \\
    \exists b_2:\ \abo(b_2,a_1,b_1)\ \wedge\ \match(a_2,b_2)\  \wedge & \exists_{D(a_2)} X_2:  \comm(a_1,b_1,X_1,a_2,b_2,X_2) \\
    \forall a_3:\ \sabo(a_3,a_2,b_2) \implies  \\ 
    \exists b_3:\ \abo(b_3,a_2,b_2)\ \wedge\ \match(a_3,b_3) \ \wedge & \exists_{D(a_3)} X_3:  \comm(a_2,b_2,X_2,a_3,b_3,X_3) \\
    \vdots \\
    \forall a_m:\  \sabo(a_m,a_{m-1},b_{m-1}) \implies \\
    \exists b_m:\ \abo(b_m,a_{m-1},b_{m-1})\ \wedge\  \match(a_m,b_m) \ \wedge & \exists_{D(a_m)} X_m: \\& \enspace \comm(a_{m-1},b_{m-1},X_{m-1},a_m,b_m,X_m)
  \end{array}
\end{equation}
In the formula: 
\begin{enumerate}
    \item $a_i$ and $b_i$ quantify over $C \cup D$.
    \item $\match(a_i,b_i)$ is true if $a_i$ and $b_i$ are in different posets and the dimensions of the associated vector spaces match.
    \item $\sabo(a_i,a_{i-1},b_{i-1})$ is true iff $a_{i-1},b_{i-1}$ are from different posets and $a_i$ is strictly greater than the element $a_{i-1}$ or $b_{i-1}$ of the matching poset. (The ``s'' stands for ``strictly''.)
    \item $\abo(b_i,a_{i-1},b_{i-1})$ is defined in the same way, but we only require that $b_i$ is greater than the element of the matching poset.
    \item $\comm(a_{i-1},b_{i-1},X_{i-1},a_i,b_i,X_i)$ is true if $a_{i-1}$ and $b_{i-1}$ are in different posets as well as $a_i$ and $b_i$. Furthermore, the matrices $X_{i-1}$ and $X_i$ have to be chosen such that the diagram induced by the four elements commutes. Here we assume that the matrices map from $F$ to $G$. 
    \item $D(a_i)$ is the dimension of the associated vector space, which is either $F(a_i)$ or $G(a_i)$.
    \item $\exists_D X$ quantifies over $\F$-matrices of size $D \times D$.
    \item $m = |C| + |D|$.
\end{enumerate}

A model of the formula $\Phi$ in (\ref{eq:Phi}) can be thought of as a tree. The root has a child for every $a_1 \in C \cup D$. The node corresponding to each $a_1$ is labeled with a triple $(a_1,X_1,b_1)$  such that $b_1$ is from the other poset with a matching dimension and $X_1 \in \F^{D(a_1) \times D(a_1)}$ is an invertible matrix. For each element $a_2$ that is strictly above the matching element $a_1$ or $b_1$, we select an element $b_2$ that is above the matching element $a_1$ or $b_1$ such that $a_2$ and $b_2$ are in different posets and the dimensions of the associated vector spaces match. Finally, we select an invertible matrix $X_2$ such that the diagram induced by the poset elements $a_1, a_2, b_1, b_2$ as well as the two isomorphisms $X_1$ to $X_2$ commutes. $(a_2,X_2,b_2)$ becomes the label of the new nodes. We go on recursively like this until both poset elements of $(a_i,X_i,b_i)$ are maximal. This happens at the latest when $i = m = |C| + |D|$. Note that when we have reached maximal elements, the predicate $\sabo(\dots)$ is always false, therefore, the implication is always true in this case. Thus it does not matter when we reach a maximal pair of poset elements earlier than after $m$ steps.

\begin{lemma}
    When the formula $\Phi$ has a model, then the diagrams $F$ and $G$ are bisimilar.
\end{lemma}

\begin{proof}
 We claim that the set $S$ of labels of the tree $T$ constructed above forms a bisimulation. For every $c \in C$ there exists $d \in D$ with $F(c) = G(d)$ and an $M \in \F^{F(c) \times F(c)}$ such that $(c,M,d)$ is a label. This is already ensured by the first layer of nodes in $T$. The same is true for the symmetric statement with the roles of $c$ and $d$ swapped. Finally, if there is a tuple $(c,M,d) \in S$, then for any $c < c'$, there must be a $d' \ge d$ such that $F(c') = G(d')$ and an invertible matrix $M' \in \F^{F(c') \times F(c')}$ such that $M'F(c \le c') = G(d \le d')M$. This is ensured by the $\exists b_i \dots$ part of the formula. The same is true for the symmetric statement for any $d < d'$. Since $c < c'$ or $d < d'$, we can stop at depth $m$, since any ascending chain in $C \times D$ with the order $(c,d) \le (c',d')$ iff $c \le c' \wedge d \le d'$ has length at most $|C| + |D| = m$. 
\end{proof}

We can also prove the converse.

\begin{lemma}
    If $F$ and $G$ are bisimilar, then $\Phi$ has a model. 
\end{lemma}

\begin{proof}
    Let $S$ be a bisimulation. For each $c \in C$, we choose a tuple $(c,M,d) \in S$ and label one child of the root with it. Such tuples exist by the definition of bisimulation. In the same way for every $d \in D$, we choose a tuple $(c,M,d)$ and add a child to the root. For every child with label $(c,M,d)$ such that either $c$ or $d$ are not maximal, we add for each $c' > c$ a child with tuple $(c',M',d')$ for some $d' \ge d$ and $M'$ such that the diagram induced by $c,c',d,d'$ commutes. Such a triple exists by the definition of bisimulation. We do the same for each $d' > d$. This obviously creates a model of $\Phi$.
\end{proof}

Now we immediately get the following theorem.

\begin{theorem} \label{alt:bisim:PSPACE}
    $\BisimFD_{\F} \in \PSPACE$ when $\F$ is a finite field.
\end{theorem}

\begin{proof}
The last two lemmas show that $F$ and $G$ are bisimilar if and only if the formula $\Phi$ in \cref{eq:Phi} has a model. Since $\Phi$ is a QBF, we can brute-force over all assignments to the quantified variables in polynomial space and therefore decide its satisfiability in $\PSPACE$.

\end{proof}


\section{Complexity of model checking for finitary diagrams}

There is a nondeterministic polynomial time reduction from $\posFF$ to $\ETIM$. Since $\ETIM \in \NP$, this proves that $\posFF \in \NP$, improving on the $\PSPACE$ upper bound by \cite{DBLP:conf/RelMiCS/Dubut20}. 

\begin{theorem}
    $\posFF_{\F} \leNP \ETIM_{\F}$ for all fields $\F$.
\end{theorem}

\begin{proof}
In the model checking problem, there are two kinds of choices that need to be made: In a formula of the form $[n]P$, we need to choose an isomorphism. This will be done by $\ETIM$. The second kind of choice is in formulas of the form $\langle M \rangle P$. Here we can choose the next element $c'$ of the diagram. This choice will be modeled by the nondeterminism of the reduction. This reduction is implicit in the work of \cite{DBLP:conf/RelMiCS/Dubut20}, when he proves the $\PSPACE$ upper bound.    

Our nondeterministic reduction will add quantifiers of the form $\exists_{n} X$ one after another to the output formula $O$ as well as linear equations. In the beginning, our input is a diagram $F$, and an object $c$ and an object formula $P$. The reduction proceeds recursively along the structure of $P$. When $P$ is a morphism formula, then besides $F$ and $c$, we will also have a variable matrix $X$ as an input. Our reduction simulates the semantic
rules of \cref{subsec:FF} as follows:
\begin{itemize}
    \item if $P = [n] S$, then if $n \not= F(c)$, we reject. Otherwise, we add the quantifier $\exists_n X$ for some fresh variable $X$ to $O$ and we go on with $F,c,X,S$.
    \item if $P = ?S$, then we go on with $F,c,S$.
    \item if $P = \top$, then we accept and output $O$.
    \item if $P = P_1 \wedge P_2$, then we first go on with $F,c,X,P_1$ and then with $F,c,X,P_2$.
    \item If $P = \langle M \rangle P'$ with $M$ being an $n_2 \times n_1$-matrix, then we first check whether $n_1 = F(c)$. If not, then we reject. Otherwise, we guess a $c'$ with $c \le c'$ such that $F(c') = n_2$. If no such $c'$ exists, then we reject. Otherwise, we add $\exists_{n_2} Y$ to $O$ for some fresh variable $Y$ as well as the equations $M \cdot X = Y \cdot F(c \le c')$. Go on with $F,c',Y,P'$.
\end{itemize}
By construction $(F,c) \in \posFF$ iff there is an accepting path on which we output a satisfiable $\ETIM$-instance.
\end{proof}

Now, using \cref{etim-np} immediately gives us the following corollary.

\begin{corollary} \label{cor:posFF:NP}
    $\posFF_{\F} \in \NP$ for all fields $\F$.
\end{corollary}

Our next goal will be to show $\NP$-hardness for $\posFF$. We will achieve this by reducing \textsc{Clique} to $\posFF$. For this reduction, we invent a gadget called constrained layered poset, which we outline below.

\begin{definition}[Constrained Layered Poset (CLP)]
    A poset $C$ is said to be a constrained layered poset (CLP) if it satisfies the following properties:
    \begin{itemize}
        \item \textbf{Layered structure.} For some $k,n \in \mathbb{N}$, we have $C = \{c_{i,j} \ |\ i \in [k], j \in [n]\}$. The elements are  ordered by $c_{i,j}\le c_{i',j'}$ iff $i<i'$ or $(i,j)=(i',j')$. Thus, the poset has $k$ layers where the $i$-th layer is an antichain $\{c_{i,1},\cdots,c_{i,n}\}$, and every element in layer $i$ is smaller than every element in layer $i'>i$.
        \item \textbf{Set-labels.} Every pair $a \leq b$ in the poset is labeled by a set $L(a,b) \subseteq U$ for some universe $U$. We will call $L: C \times C \rightarrow 2^U$ the label-function of the CLP.
        \item \textbf{Triplet criterion.} For any triplet $a\leq b\leq c$ in $C$, we have $L(a,b) \cap L(b,c) = L(a,c)$.
    \end{itemize} 
\end{definition}

Notice that a finitary diagram $F: C \rightarrow A$ can be visualized as the poset $C$ having every pair $c \leq c'$ labeled by a matrix $F(c \leq c')$. However, in CLPs defined above, the pairs $c \leq c'$ are labeled by sets $L(c, c')$ instead. The idea behind this is to only focus on finitary diagrams where each matrix $F(c \leq c')$ is a diagonal matrix with $0-1$ diagonal entries and its support being the index-set $L(c, c')$. Then the triplet criterion basically captures the fact that $F(a \leq c) = F(b \leq c) \cdot F(a \leq b)$ whenever $a \leq b \leq c$ in $C$.

In the following lemma, we take the first step in our reduction. We show that given a graph, we can construct a CLP in polynomial time satisfying certain properties capturing the edge-relations in the graph.

\begin{lemma}\label{CLP-construct}
    Given an undirected graph $G = ([n],E)$ and two integers $k,m$ with $m \geq n$, we can construct in polynomial time a CLP $C = \{c_{i,j} \ |\ i \in [k], j \in [n]\}$ with label-function $L$ such that for all $1 \leq i < i' \leq k$ and $j,j' \in [n]$, 
    $$|L(c_{i,j},c_{i',j'})| = \begin{cases}
        f(i,i',k,m) & \text{if } (j,j') \in E,\\
        f(i,i',k,m)-1 & \text{otherwise},
    \end{cases}$$
    for some appropriately chosen function $f:\Z^4 \rightarrow \Z$.
\end{lemma}

\begin{proof}
We are given an undirected graph $G=([n],E)$ and two integers $k, m$ with $m \geq n$. We have to construct a CLP satisfying the given properties.

\noindent\textbf{The poset:}
Following the definition of CLP, we can define the comparabilities in our poset $C = \{c_{i,j}\ |\ i\in[k],\ j\in[n]\}$ as follows:
$$c_{i,j}\le c_{i',j'} \text{ iff } i<i' \text{ or }(i,j)=(i',j').$$ 

\noindent\textbf{The universe:}
Let $X := \{e_{u,v},\ \be_{u,v} \ |\ (u,v)\in[m]^2 \}$ be a set of symbols and set the universe $U := [k]^2 \times X$.
We also define a map $\eta :[m]^2 \rightarrow X$ given by
$$\eta(u,v)=
\begin{cases}
e_{u,v} & \text{if } (u,v)\in E,\\
\be_{u,v} & \text{if } (u,v)\notin E,
\end{cases}
$$
which we are going to use later while constructing the labels. Note that in particular, $\eta(u,v) = \be_{u,v}$ for all $(u,v) \in [m]^2 \setminus [n]^2$ since $E \subseteq [n]^2$.

\noindent\textbf{Labels:}
We have to assign $L(c,c')\subseteq U$ for all comparable $c \leq c'$ in $C$. First of all, we set
$L(c,c):=U \text{ for all }c\in C$. Now we handle the strictly comparable pairs. For $i<i'$ and $j,j'\in[n]$, we define
$$L(c_{i,j},c_{i',j'}) := Q_{i,i'} \ \sqcup\ R_{i,i',j} \ \sqcup\ S_{i,i',j'} \ \sqcup\ T_{i,i',j,j'}$$
where 
\begin{align*}
Q_{i,i'} &= \{(a,a',x)\in [k]^2\times X \ |\ a<i<i'<a'\},\\
R_{i,i',j} &= \{(i,a',e_{j,b})\in [k]^2\times X \ |\ a'>i',\ b\in[m]\},\\
S_{i,i',j'} &= \{(a,i',\eta(b,j'))\in [k]^2\times X \ |\ a<i,\ b\in[m]\},\\
T_{i,i',j,j'} &=
\begin{cases}
\{(i,i',e_{j,j'})\} & \text{if } (j,j')\in E,\\
\emptyset & \text{otherwise.}
\end{cases}
\end{align*}
The four sets are pairwise disjoint because they have mutually exclusive constraints on the first two coordinates.

For every $i<i'$ and every $j,j'\in[n]$, we have
$$|Q_{i,i'}| = (i-1)(k-i')\cdot 2m^2,\quad |R_{i,i',j}|=(k-i')m,\quad |S_{i,i',j'}|=(i-1)m,\quad |T_{i,i',j,j'}|=\mathbbm{1}_{(j,j')\in E}\ .$$
Hence,
$$|L(c_{i,j},c_{i',j'})|
=
2(i-1)(k-i')m^2 + (k-i')m + (i-1)m + \mathbbm{1}_{(j,j')\in E} \ .$$
Now define the function $f:\Z^4 \rightarrow \Z$ as 
$$f(i,i',k,m) = 2(i-1)(k-i')m^2 + (k-i')m + (i-1)m + 1\ .$$
Then for all $j,j'$,
\begin{equation*}\label{label-size}
|L(c_{i,j},c_{i',j'})|
=
\begin{cases}
f(i,i',k,m) & \text{if }(j,j')\in E,\\
f(i,i',k,m)-1 & \text{otherwise.}
\end{cases}
\end{equation*}

\noindent\textbf{Satisfaction of the triplet criterion:}
It remains to show that for all $a\le b\le c$ in $C$,
\[
L(a,b)\cap L(b,c) = L(a,c).
\]
Assume $a=c_{i,j}$, $b=c_{i',j'}$, $c=c_{i'',j''}$ with $i<i'<i''$.
The sets $L(c_{i,j},c_{i',j'})$ and $L(c_{i',j'},c_{i'',j''})$ can be written as unions of their $Q,R,S,T$ parts and hence their intersection is the union of the cross-intersections between these parts. The only nonempty cross-intersections are:
\begin{align*}
Q_{i,i'}\cap Q_{i',i''} &= Q_{i,i''},
&&R_{i,i',j}\cap Q_{i',i''} = R_{i,i'',j},\\
S_{i',i'',j''} \cap Q_{i,i'} &= S_{i,i'',j''},
&&R_{i,i',j}\cap S_{i',i'',j''} = T_{i,i'',j,j''}.
\end{align*}
The first three equalities are easy to see. For the last equality, observe that $R_{i,i',j}$ contains tuples with first coordinate $i$ and third coordinate $e_{j,b}$ for $b \in [m]$ while $S_{i',i'',j''}$ contains tuples with second coordinate $i''$ and third coordinate $\eta(b,j'')$ for $b \in [m]$. Therefore, their intersection can contain at most one element $(i,i'',e_{j,j''})$ and it contains this element only when $\eta(j,j'') = e_{j,j''}$, i.e., when $(j,j'')\in E$. It follows that the intersection equals to $T_{i,i'',j,j''}$.
Therefore,
\[
L(c_{i,j},c_{i',j'})\cap L(c_{i',j'},c_{i'',j''})
=
Q_{i,i''}\sqcup R_{i,i'',j}\sqcup S_{i,i'',j''}\sqcup T_{i,i'',j,j''}
=
L(c_{i,j},c_{i'',j''}).
\]
\end{proof}

As we described earlier, the idea behind using the set-labels $L(c,c')$ for partial orders $c \leq c'$ in CLPs is to define a finitary diagram $F: C \rightarrow A$, where each $F(c \leq c')$ is a diagonal matrix with $0-1$ diagonal entries and its support being the index-set $L(c \leq c')$. In the following lemma, we make this idea explicit in order to lift the CLP-gadget of \cref{CLP-construct} to a finitary-diagram-gadget.
 
\begin{lemma}\label{diagram-construct}
    Given a field $\F$, an undirected graph $G = ([n],E)$ and two integers $k,m$ with $m \geq n$, we can construct in polynomial time a finitary diagram $F$ from a CLP $C = \{c_{i,j} \ |\ i \in [k], j \in [n]\}$ to a category $A$ composed of a single $\F$-vector space and linear maps, such that for all $1 \leq i < i' \leq k$ and $j,j' \in [n]$, 
    $$\rk F(c_{i,j}\leq c_{i',j'}) = \begin{cases}
        f(i,i',k,m) & \text{if } (j,j') \in E,\\
        f(i,i',k,m)-1 & \text{otherwise},
    \end{cases}$$
    for some appropriately chosen function $f:\Z^4 \rightarrow \Z$.
\end{lemma}

\begin{proof}
     Given the undirected graph $G = ([n],E)$ and the integers $k,m$ with $m \geq n$, first use \cref{CLP-construct} to construct a CLP $C = \{c_{i,j} \ |\ i \in [k], j \in [n]\}$ with label-function $L$ such that for all $1 \leq i < i' \leq k$ and $j,j' \in [n]$, 
    $$|L(c_{i,j},c_{i',j'})| = \begin{cases}
        f(i,i',k,m) & \text{if } (j,j') \in E,\\
        f(i,i',k,m)-1 & \text{otherwise},
    \end{cases}$$
    for some function $f:\Z^4 \rightarrow \Z$. Now we will define the finitary diagram $F: C \rightarrow A$.

    Let $U$ be the universe used in the label-function $L$, i.e., $U = \bigcup\limits_{\substack{a,b \in C\\a \leq b}} L(a,b),$
    and let $q := |U|$.
    We will define the range of the diagram to be the single-object category $A := \{q\}$ where the object $q$ represents the vector space $\F^q$. Therefore, $F(c) = q$ and $F(c \leq c) = \mathbf{I}_q$ for all $c \in C$.

    Choose any bijection $\phi: [q] \rightarrow U$. Now, given a subset $S \subseteq U$, define $\id_S$ to be the diagonal $q\times q$ matrix whose $i$-th diagonal entry is $1$ if $\phi(i) \in S$ and $0$ otherwise. We have 
    $$\rk \id_S = |S| \quad\text{ and }\quad \id_S \cdot \id_T = \id_{S \cap T} $$
    for all $S,T \subseteq U$. Now for all $a \leq b$ in $C$, define $$F(a \leq b) = \id_{L(a,b)}.$$ Then for all $a \leq b \leq c$ in $C$,
    $$F(b\leq c) \cdot F(a\leq b) = \id_{L(a,b)}\cdot \id_{L(b,c)} = \id_{L(a,b)\cap L(b,c)} = \id_{L(a,c)} = F(a\leq c)\ ,$$
    as desired in a diagram, and for all $1 \leq i < i' \leq k$ and $j,j' \in [n]$,
    $$\rk F(c_{i,j}\leq c_{i',j'}) = |L(c_{i,j},c_{i',j'})| = \begin{cases}
        f(i,i',k,m) & \text{if } (j,j') \in E,\\
        f(i,i',k,m)-1 & \text{otherwise}.
    \end{cases}$$
\end{proof}

 Having built the necessary gadget, we can proceed towards the main proof now. The following lemma gives a necessary and sufficient condition for a special kind of finitary formula being satisfiable. Focusing on these special kind of formulas will be sufficient for us to prove $\NP$-hardness of the general problem.

\begin{lemma}\label{model-satisfying-conditions}
    Let $F$ be a finitary diagram from a poset $C$ to a  category $A = \{q\}$, where the object $q$ represents the vector-space $\F^q$ for some field $\F$. Let $S = [q]\langle M_1\rangle\langle M_2\rangle\cdots\langle M_{k}\rangle$ be an object formula for singular matrices $M_1,\cdots,M_{k}$ over $\F$ of dimension $q \times q$. Then given $c_1 \in C$, we have $F,c_1 \models S$ if and only if there exists a chain $c_1 < c_2 < \cdots < c_{k+1}$ in $C$ satisfying 
    $$\rk(M_{i'-1} \cdots M_{i+1}M_i) = \rk F(c_i \leq c_{i'})$$
    for all $1 \leq i < i' \leq k+1$.
\end{lemma}

\begin{proof}
We have

$F,c_1 \models [q]\langle M_1\rangle\langle M_2\rangle\cdots\langle M_k\rangle$ if and only if

$ \exists$ a chain $c_1 \leq c_2 \leq \cdots \leq c_{k+1}$ in $C$ and matrices $X_1,\cdots,X_{k+1} \in \GL_q(\F)\ :\ \forall i\in [k],\\ M_i \cdot X_i = X_{i+1}\cdot F(c_i \leq c_{i+1})$.

By \cref{rank-invariant}, such matrices $X_1,\cdots,X_{k+1}$ can exist if and only if for all $1 \leq i < i' \leq k+1$,
$$\rk(M_{i'-1} \cdots M_{i+1}M_i) = \rk F(c_i \leq c_{i'})\ .$$

Since each $M_i$ is singular and $F(c\leq c) = \id_q$ for all $c \in C$, we conclude that $c_1,\cdots,c_{k+1}$ are distinct elements if the above condition is to be satisfied. Hence, we must have $c_1 < c_2 < \cdots < c_{k+1}$, as desired.
\end{proof}

 Now we are ready to prove $\NP$-hardness for $\posFF$.
 
\begin{theorem}\label{posFF-hardness}
    $\posFF_{\F}$ is $\NP$-hard for all fields $\F$.
\end{theorem}

\begin{proof}
We will give a polynomial-time reduction from \textsc{Clique} to $\posFF_{\F}$. Given a simple undirected graph $G$ on the vertex set $[n]$ and a parameter $k$, we have to decide whether $G$ has a $k$-clique. Consider the (loopless) undirected graph $G' = ([n+1], E')$ where 
    $$E' = \{(i,j) \in [n+1]^2\ |\ i=n+1 \text{ or } j=n+1 \text { or there is an edge between vertex $i$ and $j$ in $G$} \}\ .$$ Clearly, $G$ has a $k$-clique if and only if $G'$ has a $(k+1)$-clique containing the vertex $(n+1)$.
    
 Now, apply \cref{diagram-construct} on the graph $G'$ with the parameters $k+1$ and $n+1$. Then we can construct in polynomial time a finitary diagram $F$ from a CLP $C = \{c_{i,j} \ |\ i \in [k+1], j \in [n+1]\}$ to a single-object category $A$ such that for all $1 \leq i < i' \leq k+1$ and $j,j' \in [n+1]$, 
    $$\rk F(c_{i,j}\leq c_{i',j'}) = 
        f(i,i',k+1,n+1) - \mathbbm{1}_{(j,j')\not\in E'},$$
    for some function $f:\Z^4 \rightarrow \Z$.

Notice that if we applied \cref{diagram-construct} on the complete graph (with loops) on $(n+1)$ vertices instead with the same parameters, then we would obtain another finitary diagram $F'$ from the same CLP $C$ to some single-object category $A'$ such that for all $1 \leq i < i' \leq k+1$ and $j,j' \in [n+1]$, 
    $$\rk F'(c_{i,j}\leq c_{i',j'}) = 
        f(i,i',k+1,n+1) ,$$
because $(j,j')$ would always be an edge. 
We use the second diagram $F'$ to define the finitary formula in the $\posFF$-instance we create, while the first diagram $F$ will be the diagram in it. 
Let us set $M_i := F'(c_{i,1} \leq c_{i+1,1})$ for $1 \leq i \leq k$. Then for all $1\leq i < i' \leq k+1$, 
$$\rk (M_{i'-1} \cdots M_{i+1}M_i) = \rk F'(c_{i,1} \leq c_{i',1}) = f(i,i',k+1,n+1)\ .$$

Note that the single-object categories $A$ and $A'$ generated in the above two applications of \cref{diagram-construct} may be different. However, with very slight modification in the proof of the lemma, we can ensure that $A$ and $A'$ are equal to the same category $\{q\}$. Hence, each $M_i$ is a $q \times q$ matrix.

Now define the object formula $S := [q]\langle M_1\rangle\langle M_2\rangle \cdots \langle M_k\rangle$. We claim that $F,c_{1,n+1} \models S$ iff $G'$ has a $(k+1)$-sized clique containing the vertex n+1. 

Using \cref{model-satisfying-conditions}, we have $F,c_{1,n+1} \models S$ iff there exists a chain  $c_{1,b_1} = c_{1,n+1} < c_{2,b_2} < \cdots < c_{k+1,b_{k+1}}$ in $C$ (note that any chain of $k+1$ increasing elements in $C$ must pick exactly one element from each layer) such that for all $1 \leq i < i' \leq k+1$,
\begin{align*}
    &\rk(M_{i'-1} \cdots M_{i+1}M_i) = \rk F(c_{i,b_i} \leq c_{i',b_{i'}}),\\
    &\text{which means } f(i,i',k+1,n+1) = f(i,i',k+1,n+1) - \mathbbm{1}_{(b_i,b_{i'})\not\in E'},\\
    &\text{which means } (b_i,b_{i'}) \in E'.
\end{align*}

Therefore, $F,c_{1,n+1} \models S$ \\
iff there is a $(k+1)$-clique in $G'$ formed by the vertices $n+1,b_2,b_3,\cdots,b_{k+1}$\\
iff there is a $k$-clique in $G$ formed by the vertices $b_2,\cdots,b_{k+1}$.

Hence, the problem \textsc{Clique} reduces to $\posFF_{\F}$, making the latter $\NP$-hard.
\end{proof}

\section{Existential theory of special linear matrices}

The (generalized) existential theory of special linear matrices is defined in the same way as the (generalized) existential theory of invertible matrices. The only difference is that we quantify over matrices with determinant equal to $1$. Our efficient algorithm from \cref{sec:ETIM:RP} does not work in this situation, since the Schwartz-Zippel lemma can only test whether a polynomial is non-zero, but not whether there is an input at which it evaluates to $1$. In fact, we will show that the generalized existential theory of special linear matrices $\genETSM_{\Rset}$ is $\existsR$-complete.

First, we define the $\existsR$-complete problem $\ETRinv$, which we will reduce to $\genETSM_{\Rset}$.

\begin{definition}
    In $\ETRinv$, we are given formulae of the form $\exists x_1 \dots \exists x_n:\, c_1 \wedge c_2 \wedge c_3\wedge \dots \wedge c_m$, where all $c_i$ are of one of the following forms:
    \[ x=1, \enspace  x +y = z, \enspace x \cdot y = 1\]
    Here, $x,y$ and $z$ are arbitrary variables. The question is whether there exists an assignment of real nonzero values to the variables such that the given formula evaluates to true.
\end{definition}

There are many variants of $\ETRinv$ known. To the best of our knowledge, the first one was defined by \cite{DBLP:journals/jacm/AbrahamsenAM22}, see also \cite[Problem L5--L7]{DBLP:journals/corr/abs-2407-18006}. In contrast to our definition, \cite{DBLP:journals/jacm/AbrahamsenAM22} also requires that the domain of the variables is $[\frac{1}{2},2]$.

The following problem is well-known to be $\existsR$-complete: decide satisfiability of a conjunction of polynomial equations only of the form  
\begin{equation} \label{eq:tseitin}
  x = 1, \enspace x + y = z, \enspace x \cdot y = z.
\end{equation}
This follows essentially from the $\existsR$-completeness of the feasibility problem together with Tseitin's trick to decompose an arbitrary system of polynomial equations into a system of equations of the form (\ref{eq:tseitin}). Furthermore, by shifting the variables and using the famous ball theorem (see \cite{DBLP:journals/jsc/GrigorevVV88}), we can construct an equivalent system of polynomial equations of the same form such that the former system has a solution if and only if the latter has one with the absolute value of all coordinates being greater than 1.

Therefore, the only step left is to remove multiplications of the form $x \cdot y = z$ and replace them with multiplications of the form $x \cdot x' = 1$. First, we build squares:
\begin{align*}
    x \cdot x' &= 1 &&\implies x' = x^{-1} = \frac{1}{x}\\
    x_{-1} + 1 &= x &&\implies x_{-1} = x-1\\
    x_{-1} \cdot x_{-1}' &= 1 &&\implies x_{-1}' = (x-1)^{-1} = \frac{1}{x-1}\\
    t + x' &= x_{-1}' &&\implies t = \frac{1}{x-1}-\frac{1}{x}\\
    t \cdot t' &=1 &&\implies t'= \left( \frac{1}{x-1}-\frac{1}{x} \right)^{-1}\\
    t' + x &= X &&\implies X = \left( \frac{1}{x-1}-\frac{1}{x} \right)^{-1} + x = \left( \frac{1}{x^2 - x} \right)^{-1} + x = x^2.
\end{align*}
If the original ETR instance has a solution, then it has one with all variables having an absolute value greater than $1$. Therefore, all internal variables in the above reduction gadget will never assume the value $0$, hence they form a valid solution for $\ETRinv$.  

And then we can construct multiplication similarly:
    \[ \frac{(x+y)^2 - x^2 -y^2}{2} = x \cdot y\ .\]
So, we have shown that multiplication can be expressed in $\ETRinv$ with a constant amount of additional space, which completes the reduction. Thus, we obtain

\begin{proposition}
    $\ETRinv$ is $\existsR$-complete.
\end{proposition}

Now we are ready to prove $\existsR$-hardness of $\genETSM_{\Rset}$. In fact, our proof will show that this is true even if we restrict ourselves to two kinds of linear equations in $\genETSM_{\Rset}$, namely matrix equations of the form $XA = BY$ and $XA = YB$. It remains an interesting open question whether only the first type is sufficient to prove $\existsR$-hardness, that is, whether $\ETSM_{\Rset}$ is already $\existsR$-complete.

\subsection{Representing variables} \label{subsec:var}
We store variables using $2 \times 2$-matrices $\left(\begin{smallmatrix} x & r_1 \\ r_2 & x' \end{smallmatrix}\right)$. By adding the constraint 
\begin{align*}
    \begin{pmatrix}
        x & r_1\\
        r_2 & x'
    \end{pmatrix}
    \begin{pmatrix}
        1 & 0\\
        0 & 0
    \end{pmatrix}
    &=
    \begin{pmatrix}
        1 & 0\\
        0 & 0
    \end{pmatrix}
    \begin{pmatrix}
        x & r_1\\
        r_2 & x'
    \end{pmatrix},
\end{align*}
which is equivalent to 
\begin{align*}
    \begin{pmatrix}
        x & 0\\
        r_2 & 0
    \end{pmatrix}
    &=
    \begin{pmatrix}
        x & r_1\\
        0 & 0
    \end{pmatrix},
\end{align*}
we ensure that $r_1 = r_2 = 0$. $x$ and $x'$ are not constrained by this equation. However, since we quantify over matrices of determinant $1$, we get the additional constraint $x x' = 1$. This ensures that $x$ can never be zero and also forces that $x'$ is the inverse of $x$. 

\subsection{Setting variables to $1$} \label{subsec:one}

We first show how to simulate equations of the type $x = 1$. For this, we have a special $1 \times 1$-matrix $E$. Since we quantify only over matrices of determinant $1$, simply quantifying over $E$ ensures $E = (1)$. Now we take the $2 \times 2$-matrix $\left(\begin{smallmatrix} x & 0 \\ 0 & x'\end{smallmatrix}\right)$
that stores $x$ and add the constraint:
\begin{align*}
        E
    \begin{pmatrix}
        1 & 0
    \end{pmatrix}
    &=
    \begin{pmatrix}
        1 & 0
    \end{pmatrix}
    \begin{pmatrix}
        x & 0\\
        0 & x'
    \end{pmatrix}\\
  \iff  \begin{pmatrix}
        1 & 0
    \end{pmatrix}
    &=
    \begin{pmatrix}
        x & 0
    \end{pmatrix}
\end{align*}
This forces $x = 1$, and henceforth $x' = 1$.

\subsection{Equations of the form $xy = 1$} \label{subsec:prod}

Given the two representations of $x$ and $y$, $\left(\begin{smallmatrix} x & 0 \\ 0 & x'\end{smallmatrix}\right)$ and $\left(\begin{smallmatrix} y & 0 \\ 0 & y'\end{smallmatrix}\right)$, we need to ensure that $x = y'$, since we already know that $xx' = 1$ (and $yy' = 1$). This is achieved by the following constraint:
\begin{align*}
    \begin{pmatrix}
        x & 0\\
        0 & x'
    \end{pmatrix}
    \begin{pmatrix}
        0 & 1\\
        1 & 0
    \end{pmatrix}
    &=
    \begin{pmatrix}
        0 & 1\\
        1 & 0
    \end{pmatrix}
    \begin{pmatrix}
        y & 0\\
        0 & y'
    \end{pmatrix}\\
    \iff \begin{pmatrix}
        0 & x\\
        x' & 0
    \end{pmatrix}
    &=
    \begin{pmatrix}
        0 & y'\\
        y & 0
    \end{pmatrix}\ .
\end{align*}
Note that the two equations $xx' = 1$ and $yy' = 1$ then also enforce $y = x'$, so the second constraint is automatically fulfilled.

\subsection{Equations of the form $x + y = z$} \label{subsec:add}

The tricky part is to simulate the additions. For each addition, we have an extra variable $r$, for which we set up a $2 \times 2$-matrix as above. Second, for each addition, we have a $4 \times 4$-matrix $H = (h_{i,j})$ over which we will quantify. The difficult part is to cope with the constraint that the determinant of $H$ has to be $1$.

First we set up the equation
\begin{align*}
    \begin{pmatrix}
        h_{1,1} & h_{1,2} & h_{1,3} & h_{1,4} \\
        h_{2,1} & h_{2,2} & h_{2,3} & h_{2,4} \\
        h_{3,1} & h_{3,2} & h_{3,3} & h_{3,4} \\
        h_{4,1} & h_{4,2} & h_{4,3} & h_{4,4} 
    \end{pmatrix}
    \begin{pmatrix}
        1 & 0\\
        0 & 0\\
        0 & 0\\
        0 & 0
    \end{pmatrix}
    &=
    \begin{pmatrix}
        1 & 0\\
        0 & 0\\
        0 & 0\\
        0 & 0
    \end{pmatrix}
    \begin{pmatrix}
        r & 0 \\
        0 & r'
    \end{pmatrix}\\
   \iff \begin{pmatrix}
        h_{1,1} & 0\\
        h_{2,1} & 0\\
        h_{3,1} & 0\\
        h_{4,1} & 0
    \end{pmatrix}
    &=
    \begin{pmatrix}
        r & 0\\
        0 & 0\\
        0 & 0\\
        0 & 0 
    \end{pmatrix}.
\end{align*}
This constrains the first column of the matrix $H$ and leaves all other columns unconstrained. In a similar way, we constrain the second column:
\begin{align*}
    \begin{pmatrix}
        h_{1,1} & h_{1,2} & h_{1,3} & h_{1,4} \\
        h_{1,2} & h_{2,2} & h_{2,3} & h_{2,4} \\
        h_{3,1} & h_{3,2} & h_{3,3} & h_{3,4} \\
        h_{4,1} & h_{4,2} & h_{4,3} & h_{4,4} 
    \end{pmatrix}
    \begin{pmatrix}
        0 & 0\\
        1 & 0\\
        0 & 0\\
        0 & 0
    \end{pmatrix}
    &=
    \begin{pmatrix}
        1 & 0\\
        1 & 0\\
        0 & 0\\
        0 & 0
    \end{pmatrix}
    \begin{pmatrix}
        x & 0 \\
        0 & x'
    \end{pmatrix}\\
 \iff   \begin{pmatrix}
        h_{1,2} & 0\\
        h_{2,2} & 0\\
        h_{3,2} & 0\\
        h_{4,2} & 0
    \end{pmatrix}
    &=
    \begin{pmatrix}
        x & 0\\
        x & 0\\
        0 & 0\\
        0 & 0 
    \end{pmatrix}.
\end{align*}
By adding two similar constraints for the third and the fourth column involving $y$ and $z$, respectively, we can achieve that $H$ can only take the form
\begin{align} \label{eq:H}
    H & = \begin{pmatrix}
        r & x & y & z \\
        0 & x & 0 & 0 \\
        0 & 0 & y & 0 \\
        0 & 0 & 0 & z 
    \end{pmatrix}
\end{align}
The reason for this rather involved set of equations is that it ensures that $\det H = 1$ can be achieved.
\begin{observation}
\begin{enumerate}
    \item $\det H = rxyz$.
    \item In any solution to the equations constructed so far, $x \not = 0$, $y \not= 0$, and $z \not= 0$.
    \item Since $r$ does not appear anywhere else, we can set $r = 1/(xyz)$ to achieve $\det H = 1$. 
\end{enumerate}
\end{observation}
We set up a second matrix $T$ similar to $H$. We can take the same variable $r$ and also take the same set of equations, but we transpose each equation on both sides. Since the $2 \times 2$-matrices representing the variables are diagonal, they equal their transpose. Therefore, in any solution to the equations, $T$ will be the transpose of $H$, i.e.,
\begin{align*}
    T = H^T = \begin{pmatrix}
        r & 0 & 0 & 0 \\
        x & x & 0 & 0 \\
        y & 0 & y & 0 \\
        z & 0 & 0 & z 
    \end{pmatrix}.
\end{align*}
In particular, $\det T = \det H = 1$. Finally, we set up the equation that simulates the addition using the matrices $H$ and $T$. It is the only equation that is of the form $XA = YB$.
\begin{align} \label{eq:add}
    \begin{pmatrix}
        r & x & y & z \\
        0 & x & 0 & 0 \\
        0 & 0 & y & 0 \\
        0 & 0 & 0 & z 
    \end{pmatrix}
    \begin{pmatrix}
        0\\1 \\1\\-1 
    \end{pmatrix}
    &=
    \begin{pmatrix}
        r & 0 & 0 & 0 \\
        x & x & 0 & 0 \\
        y & 0 & y & 0 \\
        z & 0 & 0 & z
    \end{pmatrix}
        \begin{pmatrix}
        0\\1 \\1\\-1 
    \end{pmatrix}\\
\iff    \begin{pmatrix}
        x+y-z\\x \\y \\-z
    \end{pmatrix}
    &=
    \begin{pmatrix}
        0\\x \\y \\-z 
    \end{pmatrix}\ . \notag
\end{align}
The equality of the first entries of the resulting vectors simulates the addition. Note that the equations of the other three entries are trivially satisfied, so they do not impose any new constraints on our variables.

\begin{theorem} \label{thm:genETSM}
    $\ETRinv \leP \genETSM_{\Rset}$.
\end{theorem}

\begin{proof}
Assume that $\exists x_1 \dots \exists x_k\,: c_1 \wedge \dots \wedge c_m$ is a yes-instance of $\ETRinv$. Let $\xi_1,\dots,\xi_k$ be a satisfying assignment. For each variable $x_i$, we set up a $2 \times 2$-matrix as in \cref{subsec:var}. If we set the diagonal entries to $\xi_i$ and $\xi_i^{-1}$ and the off-diagonal entries to $0$, then the determinant is $1$ and the equations in \cref{subsec:var} are satisfied. If we have an equation $x_i = 1$, then $\xi_i = 1$ and the $2 \times 2$-matrix corresponding to $x_i$ satisfies the equations of \cref{subsec:one} by construction. If we have an equation $c_\mu$ of the form $x_i x_j = 1$, then $\xi_i = 1/\xi_j$ and the equations of \cref{subsec:prod} are again satisfied by construction. Finally, if we have an equation $c_\mu$ of the form $x_i + x_j = x_h$, then we first set the values of $H$ as depicted in (\ref{eq:H}) (substituting $\xi_i, \xi_j, \xi_h$ for $x,y,z$). Since $\xi_1,\dots,\xi_k$ is a feasible solution, all of them are nonzero. Therefore, we can choose $r$ in such a way that the determinant of $H$ is $1$. $r$ is only used for this addition. The entries of the $2 \times 2$-matrix for $r$ are set as above. In the same way, we can choose the values for $T$. Since $\xi_i + \xi_j = \xi_h$ also the equations $(\ref{eq:add})$ is satisfied. Thus we have a yes-instance for $\genETSM_{\Rset}$.

Conversely, assume that the $\genETSM_{\Rset}$-instance constructed above is a yes-instance. We claim that the values of the $(1,1)$-entries $\xi_i$ of the $2 \times 2$-matrices corresponding to $x_1,\dots,x_k$ are a satisfying solution to the $\ETRinv$-instance. All $\xi_i$ are nonzero, since the determinants are $1$. The equations $c_1,\dots,c_m$ are all satisfied, since for each $c_\mu$ we set up a gadget that ensures this.

The reduction is obviously polynomial-time computable.
\end{proof}

\begin{corollary} \label{cor:genETSM}
    $\genETSM_{\Rset}$ is $\existsR$-complete.
\end{corollary}



\section{Conclusions}

In the present work, we settled the complexity of the model checking problem for finitary formulae ($\NP$-complete) and significantly improved the complexity of deciding bisimilarity in finitary diagrams ($\NEXP$ in general and $\PSPACE$ for finite fields).
We gave an efficient randomized algorithm for the (generalized) existential theory of invertible matrices over infinite fields, in particular over the reals. It is an interesting, but very difficult question whether we can derandomize this algorithm for $\genETIM_{\Rset}$, since it is equivalent to the complement of symbolic determinant identity testing. Is the subproblem $\ETIM_{\Rset}$ equivalent to the complement of some identity testing problem?  Or is there an efficient deterministic algorithm? In contrast to $\genETIM_{\Rset}$, we proved that the generalized existential theory of special linear matrices is $\existsR$-complete. Is the existential theory of special linear matrices $\existsR$-hard, too?

\bibliography{reference}

\begin{appendix}

\section{Representation of quivers and Gabriel's theorem}\label{gabriel}

We give a brief introduction to Gabriel's theorem and its relation to 
Theorem~\ref{rank-invariant}. 
For more details on representations of quivers, we refer to \cite{schiffler2014quiver,derksen2017introduction}.

\begin{definition}[Quiver]
A \emph{quiver} is a quadruple $Q = (Q_0, Q_1, s, t)$, where
\begin{itemize}
    \item $Q_0$ is the set of vertices,
    \item $Q_1$ is the set of arrows,
    \item $s : Q_1 \to Q_0$ assigns to each arrow its \emph{source},
    \item $t : Q_1 \to Q_0$ assigns to each arrow its \emph{target}.
\end{itemize}
For an arrow $a \in Q_1$ with $s(a) = i$ and $t(a) = j$, we write $a : i \to j$.
\end{definition}

A quiver is nothing but a directed, potentially infinite multigraph. The notation above is the notation typically used in representation theory.

\begin{definition}[Representation of a quiver]
Let $k$ be a field.  
A \emph{representation} of a quiver $Q$ over $k$ consists of:
\begin{itemize}
    \item for each vertex $i \in Q_0$, a finite-dimensional $k$-vector space $V(i)$,
    \item for each arrow $a : i \to j$ in $Q_1$, a $k$-linear map $V(a) : V(i) \to V(j)$.    
\end{itemize}
\end{definition}

If the quiver is a simple directed acyclic graph, then a representation of a quiver can be viewed as a finitary diagram. If there are two directed path from $i$ to $j$ in a quiver, then the product of the morphisms along these path need not be the same. This is however required in a finitary diagram. But when this is not the case, that is, the quiver is a tree after disregarding directions, then a representation is indeed a finitary diagram. To prove Theorem~\ref{rank-invariant}, we will look at path quivers, which satisfy this property. 

\begin{definition}[Morphisms of representations]
Let $V$ and $W$ be two representations of $Q$.  
A \emph{morphism of representations}
\[
\varphi : V \to W
\]
is a collection of linear maps
\[
\varphi_i : V(i) \to W(i), \quad \text{for all } i \in Q_0,
\]
such that for every arrow $a : i \to j$ the following diagram commutes:
\[
\begin{CD}
V(i) @>{V(a)}>> V(j) \\
@V{\varphi_i}VV         @VV{\varphi_j}V \\
W(i) @>{W(a)}>> W(j)
\end{CD}
\]
i.e., $W(a) \circ \varphi_i = \varphi_j \circ V(a)$.
\end{definition}

\begin{theorem}[Gabriel]
Let $k$ be an algebraically closed field, and let $Q$ be a finite quiver 
without oriented cycles. Then the following are equivalent:
\begin{enumerate}
    \item The quiver $Q$ has 
    finitely many indecomposable representations up to isomorphism.
    \item The underlying undirected graph of $Q$ is a Dynkin diagram of type 
    $A_n$, $D_n$, $E_6$, $E_7$, or $E_8$.
\end{enumerate}
\end{theorem}

Dynkin diagrams are undirected graphs that encode the structure of a so-called root system. Here is it only important that $A_n$ is the undirected path of length $n$. 
Consider the orientation of $A_n$ of the form
\[
1 \longrightarrow 2 \longrightarrow \cdots \longrightarrow n.
\]
Gabriel's theorem in its full generality only holds for algebraically closed fields. But it is true over any field $\F$ for path quivers $A_n$. In particular, every indecomposable representation of $A_n$
is an \emph{interval representation}
\[
V_{[i,j]}, \qquad 1 \le i \le j \le n,
\]
defined as follows:
\begin{align*}
V_{[i,j]}(k) &= 
\begin{cases}
\mathbb{F}, & \text{if } i \le k \le j,\\[4pt]
0,          & \text{otherwise},
\end{cases} \\
V_{[i,j]}(k \to k+1) &= \mathrm{id}_{\mathbb{F}}
\quad \text{for } i \le k < j.
\end{align*}
Every representation $V$ of $A_n$ over $\mathbb{F}$ decomposes uniquely,
up to isomorphism and permutation of summands, as a direct sum of interval
representations:
\[
V \;\cong\; \bigoplus_{1 \le i \le j \le n}
    V_{[i,j]}^{\, m_{i,j}},
\]
for uniquely determined multiplicities $m_{i,j} \in \mathbb{N}$. From the uniqueness of the representation, Theorem~\ref{rank-invariant} follows.

\end{appendix}

\end{document}